\documentclass[11pt]{article}
\usepackage[a4paper,margin=1in]{geometry}
\usepackage{amsmath,amssymb,amsthm,mathtools}
\usepackage{mathrsfs}
\usepackage{bm}
\usepackage{hyperref}
\hypersetup{hidelinks}
\usepackage{cleveref}
\usepackage{microtype}
\usepackage{booktabs}
\usepackage{tikz}
\usetikzlibrary{arrows.meta,positioning}
\usepackage{pgfplots}
\pgfplotsset{compat=1.18}

\DeclareMathOperator{\MI}{I}
\newcommand{\E}{\mathbb{E}}
\newcommand{\R}{\mathbb{R}}
\newcommand{\T}{\mathsf{T}}
\newcommand{\Xb}{\mathsf{X}}
\newcommand{\Cb}{\mathsf{C}}

\newcommand{\Z}{\mathsf{Z}}
\newcommand{\Ph}{\varphi}

\newtheorem{theorem}{Theorem}

\newtheorem{lemma}[theorem]{Lemma}
\newtheorem{corollary}[theorem]{Corollary}
\theoremstyle{definition}
\newtheorem{definition}[theorem]{Definition}
\newtheorem{remark}[theorem]{Remark}

\title{Source-Side Sufficiency for the\\
Information Bottleneck: Exact Reduction\\
and Finite-Block Equivalence}
\author{Joss Armstrong\\
Ericsson Ireland, Athlone, Ireland\\
\texttt{joss.armstrong@ericsson.com}}
\date{August 2026}

\begin{document}
\maketitle

\begin{abstract}
The input side of the Information Bottleneck may carry variation that
is irrelevant to the task yet still costs rate. We identify this cost
exactly. Let $\T$ be a source, $\Cb$ a relevance variable, and
$Z=\Ph(\T)$ a deterministic statistic satisfying
$\Cb\leftrightarrow Z\leftrightarrow\T$. Encoding $Z$ in place of
$\T$ can clearly do no better. We prove it does no worse, and we
account for the difference. Averaging any encoder $p(\Xb\mid\T)$ over
the fibres of $\Ph$ preserves $\MI(\Xb;\Cb)$ and lowers the rate by
exactly $\MI(\Xb;\T\mid Z)$, and pullback reverses the map. The IB
curves and Lagrangian infima therefore coincide on standard Borel
spaces at every tradeoff parameter. The attained optima also
correspond. Every full-source minimiser factors through $Z$, and
every reduced minimiser pulls back. When $\Cb$ is finite and
distortion is logarithmic loss, the equivalence is operational.
Replacing $\T^n$ by $Z^n$ preserves the optimal remote distortion at
every blocklength and message budget, not only in the single-letter
limit. Consequences include the closed-form Gaussian IB solution as a
direct corollary, an exact account of the deterministic-relevance
pathologies, and the replacement of IB optimisation over a rich
source by optimisation over the typically much smaller statistic.
\end{abstract}

\section{Introduction}
\label{sec:intro}

The IB principle of Tishby, Pereira and
Bialek~\cite{tishby1999information} studies representations that
retain information about a target variable $\Cb$ and compress a
source variable $\T$. In its Lagrangian form, the problem is to find
an encoder $p(\Xb\mid\T)$ that minimises
\begin{equation}
\mathcal{L}_\beta(\Xb) \;=\; \MI(\Xb;\T) - \beta\, \MI(\Xb;\Cb),
\qquad \beta \ge 0.
\label{eq:ib-lagrangian}
\end{equation}
Equivalently, the problem is to maximise $\MI(\Xb;\Cb)$ subject to a
constraint on $\MI(\Xb;\T)$. The resulting IB curve is a basic description of the
relevance--compression tradeoff for the pair $(\T,\Cb)$.

In many applications, however, the source is not presented directly to
the information-limited encoder. A system first maps a rich source $\T$
to a task-specific summary $Z=\Ph(\T)$, then communicates or stores a
representation of $Z$. This raises a structural question. When does the
summary preserve the entire relevance--compression tradeoff rather than
merely one predictor or one operating point? The answer must be
task-relative. A summary adequate for one prediction target need not
retain what is required for a different label or a detection decision.

The condition used here is source-side sufficiency. The conditional
$p(\Cb\mid\T)$ must depend on $\T$ only through $Z=\Ph(\T)$, so that
$\Cb \leftrightarrow \Ph(\T) \leftrightarrow \T$ forms a Markov
chain. For any encoder $p(\Xb\mid\T)$, its fibre average defines an
encoder $q(\Xb\mid Z)$ that preserves $\MI(\Xb;\Cb)$ and satisfies
\begin{equation*}
\MI_p(\Xb;\T)-\MI_q(\Xb;Z)=\MI_p(\Xb;\T\mid Z).
\end{equation*}
The average therefore removes precisely the rate assigned to
differences among source points within a fibre of $\Ph$. It does not
change relevance.

This identity yields equality of the IB curves and Lagrangian infima
at every $\beta$ and characterises attained optima by pullback. The
reduction holds on arbitrary standard Borel spaces, without
finite-alphabet hypotheses. For finite $\Cb$ under logarithmic loss,
it also holds at finite blocklength. The best remote distortion is the
same at every blocklength and message budget. The operational
rate--distortion functions therefore coincide. We refer to these
results collectively as the \emph{IB reduction theorem} and its
finite-block operational form.

The closest prior statement is due to Harremo\"{e}s and
Tishby~\cite{harremoes2007ib}. In 2007 they observed, in the course of an
argument about distortion measures, that an input-side sufficient
statistic should preserve the finite-alphabet IB rate--distortion curve.
Their projection argument is brief and was not developed into a general
theorem. The present paper supplies the rate-removal identity on standard
Borel spaces, the variational and optimiser consequences, and the
finite-block operational equivalence.

When $Z=\Ph(\T)$ is both a deterministic function of $\T$ and
sufficient for $\Cb$, the signals $\T$ and $Z$ are
Blackwell-equivalent experiments about $\Cb$
\cite{blackwell1953equivalent}. Blackwell equivalence alone does not
compare the rate coordinates $\MI(\Xb;\T)$ and $\MI(\Xb;Z)$. The
fibre-average identity is the rate-sensitive step. Kolchinsky's recent
IB-type use of Blackwell order concerns a distinct multi-source
partial-information problem~\cite{kolchinsky2024redundancy}.

The IB problem admits familiar solution methods in two regimes. When
both $\T$ and $\Cb$ are discrete and finitely supported, the
Blahut--Arimoto fixed-point iteration is available
\cite{blahut1972computation,arimoto1972algorithm}. When the pair is
jointly Gaussian, Chechik, Globerson, Tishby and
Weiss~\cite{chechik2005gaussian} give a closed-form solution in terms
of canonical correlation structure. Between these regimes lies the
generic high-dimensional case, where estimation or optimisation of the
full joint $p(\T,\Cb)$ is difficult. A prominent practical response
is the variational IB (VIB) framework of Alemi et
al.~\cite{alemi2017deep}, which replaces the exact objective by
amortised variational bounds at the cost of surrogate objectives and
a known posterior-collapse pathology~\cite{alemi2018fixing}.

Sufficiency in this sense is a joint property of the summary and the
named task variable, and it holds exactly in a range of canonical
settings. The posterior statistic is always sufficient. Structured
conditionals that depend on a finite-dimensional parameter,
deterministic tasks that factor through the summary, and sensor sources
whose nuisance component is irrelevant given the retained signal supply
further exact cases (\Cref{sec:examples}). Exact sufficiency must be
checked separately for each task. If only approximate sufficiency is
known, the exact equalities proved here do not follow.

The linear-Gaussian and deterministic-target cases follow by direct
specialisation. A finite binary example displays the rate removed by
the fibre average exactly.

The remainder of the paper establishes the reduction formally
(\Cref{sec:setup,sec:main}), proves it (\Cref{sec:proof}), derives
the Gaussian specialisation (\Cref{sec:gaussian}), recasts the
reduction under log-loss remote source coding
(\Cref{sec:operational}), certifies the sufficiency condition in
canonical settings (\Cref{sec:examples}), gives an exact finite example
(\Cref{sec:numerics}), and discusses implications
(\Cref{sec:discussion}).

\section{Setup and notation}
\label{sec:setup}

Let $(\T,\Cb)$ be random variables on a joint probability space. Let
$\T$ take values in a standard Borel space
$(\mathcal{T},\mathcal{B}_\T)$ and let $\Cb$ take values in a standard
Borel space $(\mathcal{C},\mathcal{B}_\Cb)$. See
\cite{kechris1995descriptive} for background on standard Borel
spaces. We assume
$\MI(\T;\Cb) < \infty$. Let $\Xb$ denote a further random variable, the
IB representation, which takes values in a standard Borel space
$(\mathcal{X},\mathcal{B}_\Xb)$. It is coupled to $\T$ through a conditional
distribution $p(\Xb \mid \T)$ and satisfies the Markov chain

\begin{equation}
\Cb \;\longleftrightarrow\; \T \;\longleftrightarrow\; \Xb.
\label{eq:markov}
\end{equation}

\begin{definition}[IB curve]
\label{def:ib-curve}
The \emph{IB curve} of $(\T,\Cb)$ is the function
\begin{equation}
R \;\longmapsto\; \mathcal{I}_{\T,\Cb}(R) \;=\;
\sup\bigl\{\, \MI(\Xb;\Cb) : \, p(\Xb \mid \T) \text{ satisfies
\eqref{eq:markov} and } \MI(\Xb;\T) \le R \,\bigr\},
\qquad R \ge 0.
\end{equation}
\end{definition}

\begin{definition}[Sufficient statistic for $\Cb$ given $\T$]
\label{def:sufficient-statistic}
A measurable map $\Ph : (\mathcal{T},\mathcal{B}_\T) \to
(\mathcal{Z},\mathcal{B}_\Ph)$ is \emph{sufficient for $\Cb$ given
$\T$} if the conditional distribution $p(\Cb \mid \T)$ is a
measurable function of $\Ph(\T)$, equivalently if
\begin{equation}
\Cb \;\longleftrightarrow\; \Ph(\T) \;\longleftrightarrow\; \T
\qquad \text{forms a Markov chain.}
\label{eq:sufficient}
\end{equation}
\end{definition}

\begin{remark}
When we need to distinguish the map from the induced random
variable, we write $Z := \Ph(\T)$.
\end{remark}

\begin{remark}
Condition~\eqref{eq:sufficient} is equivalent to
$\MI(\Cb;\T \mid \Ph(\T)) = 0$, and to the existence of a regular
conditional kernel $\kappa(c \mid z)$ such that
$p(\Cb \in A \mid \T = t) = \kappa(A \mid \Ph(t))$ almost surely.
When $p(\Cb,\T)$ is given, the conditional law
$p(\Cb \mid \T=\cdot)$, viewed as an element of the space of
probability measures on $\mathcal{C}$, is a canonical sufficient
statistic. Every other sufficient statistic determines it almost
surely. In practice a convenient sufficient statistic is often much
simpler (e.g.\ a
low-dimensional projection, a posterior vector in a classification
problem, or the parameters of an exponential family).
\end{remark}

\begin{remark}
By the data-processing inequality, sufficiency of $\Ph$ implies
$\MI(\Ph;\Cb) \le \MI(\T;\Cb) < \infty$. Consequently every
relevance term $\MI(\Xb;\Cb)$ is finite. The rate
$\MI(\Xb;\Ph)$ may still be infinite for an unrestricted encoder. All
infima below are therefore read in the extended-real sense.
\end{remark}
Throughout we write $\MI_p$ for mutual information computed under
distribution $p$. We omit the subscript when the distribution is clear
from context.

\section{The reduction theorem}
\label{sec:main}

\begin{theorem}[IB reduction]
\label{thm:main}
Let $(\T,\Cb)$ be random variables on standard Borel spaces with
$\MI(\T;\Cb) < \infty$, and let $\Ph : \mathcal{T} \to \mathcal{Z}$
be sufficient for $\Cb$ given $\T$ in the sense of
\eqref{eq:sufficient}. Write $Z := \Ph(\T)$. The following statements
hold.
\begin{enumerate}
\item[(i)] \textbf{Curve preservation.}
$\mathcal{I}_{\T,\Cb}(R) = \mathcal{I}_{Z,\Cb}(R)$ for all
$R \ge 0$.
\item[(ii)] \textbf{Lagrangian preservation.} For every
$\beta \ge 0$,
\begin{equation}
\inf_{p(\Xb\mid\T)} \bigl[\MI(\Xb;\T) - \beta\,\MI(\Xb;\Cb)\bigr]
\;=\;
\inf_{p(\Xb\mid Z)} \bigl[\MI(\Xb;Z) - \beta\,\MI(\Xb;\Cb)\bigr],
\label{eq:lagrangian-equal}
\end{equation}
where the left-hand infimum is over conditionals that satisfy
\eqref{eq:markov} and the right-hand infimum is over conditionals
that satisfy the analogous Markov chain
$\Cb \leftrightarrow Z \leftrightarrow \Xb$.
\item[(iii)] \textbf{Minimiser correspondence.} For any fixed
$\beta\geq0$, if
$p^\star(\Xb \mid Z)$ attains the right-hand infimum in
\eqref{eq:lagrangian-equal}, then the conditional
\begin{equation}
p^\star(\Xb \mid \T = t) \;:=\; p^\star(\Xb \mid Z = \Ph(t))
\label{eq:pullback}
\end{equation}
attains the left-hand infimum and has the same $\MI(\Xb;\T)$ and
$\MI(\Xb;\Cb)$ values. Conversely, every minimiser
$p(\Xb\mid\T)$ of the left-hand problem satisfies
$p(\Xb\mid\T=t) = p^\star(\Xb\mid Z=\Ph(t))$ for $p(\T)$-almost every
$t$, for some $p^\star(\Xb\mid Z)$ that minimises the right-hand
problem.
\end{enumerate}
\end{theorem}

\begin{remark}[Mutual information on standard Borel spaces]
\label{rem:mi-partition}
Mutual information is used throughout in the partition-supremum
sense of Gel'fand--Yaglom and
Pinsker~\cite{pinsker1964information}. For random variables $A$
and $B$ on standard Borel spaces $(\mathcal{A},\mathcal{B}_A)$
and $(\mathcal{B},\mathcal{B}_B)$,
\begin{equation*}
\MI(A;B) \;=\; \sup_{\mathcal{P}_A,\,\mathcal{P}_B}
\MI\bigl(\mathcal{P}_A(A);\,\mathcal{P}_B(B)\bigr),
\end{equation*}
where the supremum is over pairs of finite measurable partitions
of $\mathcal{A}$ and $\mathcal{B}$ and the right-hand side is the
elementary discrete mutual information of the partition-induced
random variables. This construction is well-posed without further
regularity, coincides with the Radon--Nikodym form
$\int \log(dp_{AB}/d(p_A \otimes p_B))\,dp_{AB}$ whenever the joint
law is absolutely continuous with respect to the product marginal,
and recovers the finite-alphabet definition when both spaces are
finite. All mutual informations appearing in
\eqref{eq:lagrangian-equal}, in items~(i)--(iii), and in the
hypothesis $\MI(\T;\Cb) < \infty$ are read in this sense.
\end{remark}

\begin{remark}
Item~(i) is the \emph{information-theoretic} content of the reduction.
No information about $\Cb$ is lost by restricting attention to
$Z=\Ph(\T)$. Item~(ii) is the \emph{variational} content. The IB
Lagrangian literally takes the same value on the two problems at
every $\beta$. Item~(iii) is the \emph{algorithmic} content. A
solver run on the reduced problem yields a solver on the original
via pullback.
\end{remark}

\begin{remark}[Attainment]
\label{rem:attainment}
On finite alphabets with a bounded representation-alphabet size
$|\mathcal{X}|$, the constraint sets in
\eqref{eq:lagrangian-equal} are compact and both objectives are
continuous, so the two infima are attained and $\inf$ may be
replaced by $\min$. In the general standard-Borel setting used
throughout the paper, attainment is not automatic. Parts~(i)--(ii)
compare encoders directly and do not require it. Item~(iii) is
conditional on the existence of a minimiser, as its statement makes
explicit.
\end{remark}

\begin{figure}[!ht]
\centering
\begin{tikzpicture}[
  node distance=1.5cm and 1.6cm,
  rv/.style={draw, circle, minimum size=0.9cm, inner sep=1pt},
  arr/.style={->, >=Latex, thick}
]
  \node[rv] (C) {$\Cb$};
  \node[rv, right=of C] (Z) {$\Ph(\T)$};
  \node[rv, right=of Z] (T) {$\T$};
  \node[rv, right=of T] (X) {$\Xb$};

  \draw[arr] (C) -- (Z);
  \draw[arr] (Z) -- (T);
  \draw[arr] (T) -- (X) node[midway, above, font=\small] {$p(\Xb\mid\T)$};

  \draw[arr, dashed] (Z) to[bend right=32]
    node[midway, below, font=\small] {$q(\Xb\mid\Ph(\T))$} (X);
\end{tikzpicture}
\caption{Markov structure and pullback. The IB chain
$\Cb \leftrightarrow \T \leftrightarrow \Xb$ factors through the
sufficient statistic $\Ph(\T)$. The full chain reads
$\Cb \leftrightarrow \Ph(\T) \leftrightarrow \T \leftrightarrow \Xb$
(top row). A reduced encoder $q(\Xb\mid\Ph(\T))$ on the reduced
source $(\Ph(\T),\Cb)$ has pullback
$\tilde{p}(\Xb \mid \T = t) := q(\Xb \mid \Ph(t))$ (dashed). It realises
the same $(\MI(\Xb;\T), \MI(\Xb;\Cb))$ pair by
\Cref{lem:extension}.}
\label{fig:f1-schematic}
\end{figure}

\begin{corollary}[Computational reduction]
\label{cor:computation}
If $Z=\Ph(\T)$ takes values in a finite set of size $K$ and $\Cb$
takes values in a finite set of size $M$, the IB curve of
$(\T,\Cb)$ is determined entirely by the $K \times M$ joint
$p(Z,\Cb)$. A Blahut--Arimoto iteration on that reduced joint has
cost independent of
$\dim \mathcal{T}$ and of the marginal $p(\T)$ beyond what is
required to estimate $p(Z,\Cb)$.
\end{corollary}

\begin{proof}[Proof of Corollary~\ref{cor:computation}]
By Theorem~\ref{thm:main}(i), the IB curve of $(\T,\Cb)$ is the IB
curve of the $K \times M$ discrete distribution $p(\Ph,\Cb)$. The
Blahut--Arimoto
algorithm~\cite{blahut1972computation,arimoto1972algorithm,tishby1999information}
for this problem maintains a $K \times |\mathcal{X}|$ encoder matrix
(where $|\mathcal{X}|$ is the representation alphabet size, which
by the support
lemma~\cite[Lemma~15.4]{csiszar2011information} may be taken to
satisfy $|\mathcal{X}| \le K+1$. See also
\cite[Appendix~C]{elgamal2011network}) and iterates two fixed-point
equations, each at cost $O(KM|\mathcal{X}|)$. The cost of one
iteration is therefore $O(K^2 M)$, independent of
$\dim\mathcal{T}$ and of the structure of $p(\T)$ beyond the
$K \times M$ matrix $p(\Ph,\Cb)$. The cost of building that matrix
from $N$ samples is $O(N)$ (one pass over the data to compute
$\Ph(T_i)$ for each sample $i$), also independent of
$\dim\mathcal{T}$ given an oracle for $\Ph$.
Blahut--Arimoto is a fixed-point method for a generally non-convex
problem. This corollary concerns the dimension and per-iteration cost
of the reduced optimisation, not a global-convergence guarantee.
\end{proof}

When $Z$ is not finite, the reduced problem remains
infinite-dimensional, and quantisation or variational approximation
is required before \Cref{cor:computation} applies.

\begin{corollary}[Degenerate case and recovery of Kolchinsky et al.]
\label{cor:kolchinsky}
Suppose $\Cb$ takes values in a finite alphabet and
$\Cb = h(\T)$ almost surely for some measurable
$h : \mathcal{T} \to \mathcal{C}$. Then $h$ is sufficient for
$\Cb$ given $\T$ in the sense of
\Cref{def:sufficient-statistic}, and
Theorem~\ref{thm:main} applies with $\Ph := h$ and
$Z = h(\T) = \Cb$. In this degenerate case, the following statements
hold.
\begin{enumerate}
\item[(a)] The IB curve of the reduced pair is the diagonal,
$\mathcal{I}_{\Cb,\Cb}(R) = \min\{R, H(\Cb)\}$ for $R \ge 0$, and
Theorem~\ref{thm:main}(i) gives
$\mathcal{I}_{\T,\Cb}(R) = \min\{R, H(\Cb)\}$ for the original
pair.
\item[(b)] The IB Lagrangian~\eqref{eq:ib-lagrangian} has a single
kink at $\beta = 1$. Its infimum is $0$ for $\beta \le 1$
(attained by the constant encoder $\Xb \equiv \mathrm{const}$)
and $(1-\beta)\,H(\Cb)$ for $\beta \ge 1$ (attained by the
deterministic labelling $\Xb = \Cb$). At $\beta = 1$, every point
on the diagonal segment in part~(a) is attained by a minimiser.
\item[(c)] The pullback~\eqref{eq:pullback} identifies the
$\Xb = \Cb$ minimiser of the reduced problem with $\Xb = h(\T)$
on the original problem.
\end{enumerate}
\end{corollary}

\begin{proof}
For sufficiency, $p(\Cb \mid \T = t)$ is the point mass at $h(t)$, so
$\Cb$ is a deterministic function of $h(\T)$. The chain
$\Cb \leftrightarrow h(\T) \leftrightarrow \T$ is trivially Markov,
and $h$ is sufficient by \Cref{def:sufficient-statistic}.

Part~(a). On the reduced pair with $Z = \Cb$, the constraint chain
$\Cb \leftrightarrow \Cb \leftrightarrow \Xb$ imposes only that
$\Xb$ depend on $\Cb$ alone, and the rate constraint
$\MI(\Xb;Z) \le R$ becomes $\MI(\Xb;\Cb) \le R$ directly. The IB
objective on the reduced pair is therefore
$\sup\{\MI(\Xb;\Cb) : \MI(\Xb;\Cb) \le R\} = \min\{R, H(\Cb)\}$.
For $R \ge H(\Cb)$ the identity encoder $\Xb = \Cb$ attains
$\MI(\Xb;\Cb) = H(\Cb)$. For $R \in [0, H(\Cb))$, the erasure
mixture that outputs $\Cb$ with probability $R/H(\Cb)$ and a
constant symbol otherwise achieves $\MI(\Xb;\Cb) = R$, so the
supremum is attained and equals $\min\{R, H(\Cb)\}$.
\Cref{thm:main}(i) transfers this value to the original pair.

Part~(b). For any encoder $p(\Xb \mid \T)$ satisfying
\eqref{eq:markov}, the data-processing inequality gives
$\MI(\Xb;\Cb) \le \MI(\Xb;\T)$, hence
\begin{equation*}
\MI(\Xb;\T) - \beta\,\MI(\Xb;\Cb)
\;\ge\; (1-\beta)\,\MI(\Xb;\Cb).
\end{equation*}
For $\beta \le 1$ the right-hand side is nonnegative and vanishes
at $\MI(\Xb;\Cb) = 0$. The constant encoder attains this.
For $\beta \ge 1$ the coefficient $(1-\beta) \le 0$, so
$(1-\beta)\,\MI(\Xb;\Cb) \ge (1-\beta)\,H(\Cb)$ follows from
$\MI(\Xb;\Cb) \le H(\Cb)$. The deterministic labelling
$\Xb = h(\T) = \Cb$ gives $\MI(\Xb;\T) = \MI(\Xb;\Cb) = H(\Cb)$ and
attains the bound. Both branches evaluate to $0$ at $\beta = 1$,
so the infimum $\beta \mapsto \inf_{p(\Xb\mid\T)} \mathcal{L}_\beta$
is $\min\{0,(1-\beta) H(\Cb)\}$ with a single kink there. At
$\beta=1$, every erasure mixture from part~(a) has
$\MI(\Xb;\T)=\MI(\Xb;\Cb)$ and is therefore a minimiser. These
mixtures attain every point on the diagonal segment.

Part~(c). Take $p^\star(\Xb \mid \Cb) := \delta_{\Cb}$ (the
identity encoder), then \eqref{eq:pullback} gives
$p^\star(\Xb \mid \T = t) = \delta_{h(t)}$, i.e.\
$\Xb = h(\T) = \Cb$ on the original problem.
\end{proof}

\begin{remark}
The corollary explains the deterministic-relevance pathologies
documented by Kolchinsky et al.~\cite{kolchinsky2019caveats} as
the pullback through $h$ of the diagonal IB curve on
$(\Cb, \Cb)$. Their observation that the IB Lagrangian's optimal
representation collapses to a constant for $\beta \le 1$ and jumps
to a deterministic labelling for $\beta > 1$ is the degeneration of
the common Lagrangian to the single kink at $\beta = 1$ recorded in
part~(b) above. At the kink, every point on the diagonal segment is
optimal. The reduction thus identifies the pathology as a structural
degeneracy of the deterministic-relevance problem and a corresponding
limitation of linear scalarisation on its non-strictly-concave
frontier.
\end{remark}

\section{Proof of Theorem~\ref{thm:main}}
\label{sec:proof}

The only preliminary result needed is the pullback identity. The proof
carries out the reverse operation directly through a fibrewise encoder
average.

\begin{lemma}[Pullback preserves Markov structure and mutual informations]
\label{lem:extension}
Let $\Ph$ be sufficient for $\Cb$ given $\T$ in the sense of
\eqref{eq:sufficient}, and let $q(\Xb \mid \Ph)$ be any conditional
distribution such that $\Cb \leftrightarrow \Ph \leftrightarrow \Xb$
is Markov under $p(\Ph,\Cb)\,q(\Xb\mid\Ph)$. Define the pullback
conditional
$$
\tilde{p}(\Xb \mid \T = t) \;:=\; q(\Xb \mid \Ph = \Ph(t)).
$$
Let $\tilde{p}(\T,\Cb,\Xb) = p(\T,\Cb)\,\tilde{p}(\Xb\mid\T)$ be
the induced joint. The following statements hold.
\begin{enumerate}
\item[(a)] $\Cb \leftrightarrow \T \leftrightarrow \Xb$ is Markov
under $\tilde{p}$, so $\tilde{p}(\Xb\mid\T)$ is a feasible encoder
for the left-hand IB problem.
\item[(b)] $\MI_{\tilde{p}}(\Xb;\T) = \MI_q(\Xb;\Ph)$.
\item[(c)] $\MI_{\tilde{p}}(\Xb;\Cb) = \MI_q(\Xb;\Cb)$, where both
sides are computed under their respective induced joints over
$(\Ph,\Cb,\Xb)$.
\end{enumerate}
\end{lemma}

\begin{proof}
The encoder is generated from $\T$ alone, so~(a) holds by definition.
It also depends on $\T$ only through $\Ph$, hence
$\Xb\perp\T\mid\Ph$. The chain rule therefore gives
\begin{equation*}
\MI_{\tilde p}(\Xb;\T)
=\MI_{\tilde p}(\Xb;\Ph)
 +\MI_{\tilde p}(\Xb;\T\mid\Ph)
=\MI_{\tilde p}(\Xb;\Ph).
\end{equation*}
The $(\Ph,\Xb)$ marginal under $\tilde p$ is
$p(\Ph)q(\Xb\mid\Ph)$, which proves~(b). Sufficiency similarly
factorises the three-variable marginal as
\begin{equation*}
\tilde p(\Cb,\Ph,\Xb)
=p(\Cb,\Ph)q(\Xb\mid\Ph).
\end{equation*}
This is exactly the reduced joint, so its $(\Cb,\Xb)$ marginal and
therefore $\MI(\Xb;\Cb)$ are unchanged, proving~(c).
\end{proof}

\begin{proof}[Proof of Theorem~\ref{thm:main}]
There is one construction behind all three claims. For any feasible
$p(\Xb\mid\T)$, define its average over each fibre of $Z=\Ph(\T)$ by
\begin{equation}
q(\Xb\in A\mid Z=z)
  := \int p(\Xb\in A\mid\T=t)\,p(\T\in dt\mid Z=z).
\label{eq:fibre-encoder}
\end{equation}
The required disintegration exists on standard Borel spaces. See
\Cref{app:measurability}. We call $q$ the fibre average of $p$.

Sufficiency and feasibility give
$\Cb\leftrightarrow Z\leftrightarrow\T$ and
$\Cb\leftrightarrow\T\leftrightarrow\Xb$. Hence, conditional on
$Z=z$, the joint law of $(\Cb,\Xb)$ under $p$ is
\[
\int p(\Cb\in dc\mid Z=z)\,
     p(\Xb\in dx\mid\T=t)\,p(\T\in dt\mid Z=z)
=p(\Cb\in dc\mid Z=z)q(\Xb\in dx\mid Z=z).
\]
Thus the $(\Cb,Z,\Xb)$ marginal under $p$ is exactly the joint induced
by $p(\Cb,Z)q(\Xb\mid Z)$. In particular,
\begin{equation}
\MI_p(\Xb;\Cb)=\MI_q(\Xb;\Cb),\qquad
\MI_p(\Xb;Z)=\MI_q(\Xb;Z).
\label{eq:fibre-preserves}
\end{equation}
Since $Z$ is a function of $\T$, the chain rule also gives
\begin{equation}
\MI_p(\Xb;\T)
=\MI_q(\Xb;Z)+\MI_p(\Xb;\T\mid Z).
\label{eq:fibre-rate}
\end{equation}
We refer to \eqref{eq:fibre-rate} as the \emph{fibre-average
identity}.

We first prove~(i). If $p$ is feasible at rate $R$ for the full
problem, \eqref{eq:fibre-preserves}--\eqref{eq:fibre-rate} show that
its fibre average is feasible at rate at most $R$ for the reduced problem and
has the same relevance. Therefore
$\mathcal I_{\T,\Cb}(R)\leq\mathcal I_{Z,\Cb}(R)$. Conversely, any
reduced encoder $q(\Xb\mid Z)$ pulls back to
$\tilde p(\Xb\mid\T=t):=q(\Xb\mid Z=\Ph(t))$. By
\Cref{lem:extension}, the pullback has exactly the same rate and
relevance, so
$\mathcal I_{\T,\Cb}(R)\geq\mathcal I_{Z,\Cb}(R)$. This proves curve
equality directly.

For~(ii), subtract $\beta$ times the first identity in
\eqref{eq:fibre-preserves} from \eqref{eq:fibre-rate} to obtain
\begin{equation}
\bigl[\MI_p(\Xb;\T)-\beta\MI_p(\Xb;\Cb)\bigr]
=\bigl[\MI_q(\Xb;Z)-\beta\MI_q(\Xb;\Cb)\bigr]
 +\MI_p(\Xb;\T\mid Z).
\label{eq:fibre-lagrangian}
\end{equation}
Every full encoder is therefore no better than its fibre average,
whereas every reduced encoder has a pullback with the same objective
value. The infima therefore prove~(ii).

For~(iii), the pullback of a reduced minimiser attains the common
infimum by \Cref{lem:extension}. Conversely, if $p$ minimises the full
problem, then \eqref{eq:fibre-lagrangian} and the equality of the two
infima force both its fibre average to minimise the reduced problem and
$\MI_p(\Xb;\T\mid Z)=0$. The latter equality is exactly
$\Xb\perp\T\mid Z$, so $p$ agrees almost surely with the pullback of
its fibre average.

\end{proof}

\begin{remark}[Per-encoder wasted-rate decomposition]
\label{rem:wasted-rate}
Let $q$ be the fibre average of a full encoder $p$. The difference
between \eqref{eq:fibre-lagrangian} and the common optimum is
\begin{equation}
\bigl[\MI_p(\Xb;\T) - \beta\,\MI_p(\Xb;\Cb)\bigr] - \mathcal{L}_\beta^{\T}
\;=\;
\Bigl[\bigl(\MI_q(\Xb;Z) - \beta\,\MI_q(\Xb;\Cb)\bigr) - \mathcal{L}_\beta^{Z}\Bigr]
   + \MI_p(\Xb;\T\mid Z),
\qquad \beta \ge 0,
\label{eq:wasted-rate}
\end{equation}
where $\mathcal L_\beta^{\T}=\mathcal L_\beta^{Z}$ by
\Cref{thm:main}(ii). Both terms on the right are nonnegative. Thus an
encoder within $\delta$ of optimal spends at most $\delta$ nats on
variation in $\T$ beyond $Z=\Ph(\T)$. This is the quantitative form of
the converse in \Cref{thm:main}(iii).
\end{remark}

\section{Gaussian specialisation and the Chechik et al. solution}
\label{sec:gaussian}

We now specialise Theorem~\ref{thm:main} to the setting of
\cite{chechik2005gaussian} and show that their closed-form Gaussian
IB solution is an immediate corollary of the reduction. We then
state an explicit nonlinear generalisation.

\subsection{The linear-Gaussian case}

Let $\T \in \R^{d_\T}$ and $\Cb \in \R^{d_\Cb}$ be jointly
Gaussian with zero mean and joint covariance
\[
\Sigma \;=\;
\begin{pmatrix}
\Sigma_{\T\T} & \Sigma_{\T\Cb} \\
\Sigma_{\Cb\T} & \Sigma_{\Cb\Cb}
\end{pmatrix},
\qquad \Sigma_{\T\T} \text{ invertible.}
\]
Then
\begin{equation}
\Cb \mid \T
\;\sim\;
\mathcal{N}\!\bigl(
\Sigma_{\Cb\T}\Sigma_{\T\T}^{-1}\T,\;
\Sigma_{\Cb\Cb} - \Sigma_{\Cb\T}\Sigma_{\T\T}^{-1}\Sigma_{\T\Cb}
\bigr),
\label{eq:gaussian-cond}
\end{equation}
so $p(\Cb\mid\T)$ depends on $\T$ only through the conditional-mean
vector
\[
\Ph(\T) \;:=\; \Sigma_{\Cb\T}\Sigma_{\T\T}^{-1}\T \;\in\; \R^{d_\Cb}.
\]
This is a linear map of rank $r :=
\mathrm{rank}(\Sigma_{\Cb\T}\Sigma_{\T\T}^{-1})
\le \min\{d_\Cb,d_\T\}$, and it
is sufficient for $\Cb$ given $\T$ in the sense
of~\eqref{eq:sufficient}. The same map is the minimum
mean-square-error estimator of $\Cb$ given $\T$.
Theorem~\ref{thm:main} therefore applies.

\begin{corollary}[Recovery of Chechik et al.]
\label{cor:chechik}
In the linear-Gaussian setting above, the following statements hold.
\begin{enumerate}
\item[(a)] The IB curve $\mathcal{I}_{\T,\Cb} =
\mathcal{I}_{\Ph(\T),\Cb}$ is the IB curve of the
$r$-dimensional pair $(\Ph(\T),\Cb)$, where
$r \le \min\{d_\Cb,d_\T\}$.
\item[(b)] Within the linear-Gaussian encoder class of
\textup{\cite{chechik2005gaussian}}, a reduced encoder pulls back to
a linear map on $\T$ of rank at most $r$. This recovers the effective
rank bound in their solution.
\item[(c)] The closed-form Gaussian IB calculation of
\textup{\cite{chechik2005gaussian}} gives the same optimum for the
reduced Gaussian pair $(\Ph(\T),\Cb)$ and the full pair $(\T,\Cb)$.
\end{enumerate}
\end{corollary}

\begin{proof}
Part~(a) is Theorem~\ref{thm:main}(i) applied to $\Ph$. For~(b),
write the conditional-mean map as $P\T$ with $\operatorname{rank}P=r$.
A linear reduced encoder with signal map $B$ pulls back to the signal
map $BP$ on $\T$, whose rank is at most $r$. Equality of its objective
value follows from Theorem~\ref{thm:main}(iii). For~(c), the pair
$(P\T,\Cb)$ is jointly Gaussian. Chechik et al.'s Gaussian
eigenvalue calculation applies directly to this reduced pair, and
Theorem~\ref{thm:main}(ii) transfers its value and pullback minimiser
to the full pair.
\end{proof}

\begin{remark}
\label{rem:chechik-wasting}
The reduction supplies a conceptual explanation of the effective
rank in the linear-Gaussian solution. The IB problem itself lives in
the $r$-dimensional image of $\Ph$, and any dependence on $\T$ beyond
this image spends rate without changing relevance. The linear rank
bound was originally established
in~\cite{chechik2005gaussian} via first-order analysis of the
Gaussian Lagrangian. \Cref{thm:main} re-derives it as a structural
consequence of input-side sufficiency.
\end{remark}

\subsection{Nonlinear-Gaussian generalisation}

The reduction does not require joint Gaussianity. Suppose only that
\begin{equation}
\Cb \mid \T
\;\sim\;
\mathcal{N}\!\bigl(\mu(\Ph(\T)),\;\Sigma_\Cb(\Ph(\T))\bigr)
\label{eq:nonlinear-gaussian}
\end{equation}
for some measurable $\Ph : \mathcal{T} \to \R^k$, where $\mu$ and
$\Sigma_\Cb$ are measurable functions of $\Ph(\T)$ alone. Then $\Ph$
is sufficient for $\Cb$ given $\T$, and Theorem~\ref{thm:main}
reduces the IB problem for $(\T,\Cb)$ to the IB problem for
$(\Ph(\T),\Cb)$, which involves only the $k$-dimensional statistic
$\Ph(\T)$ regardless of $\dim\mathcal{T}$.

\begin{remark}
Form~\eqref{eq:nonlinear-gaussian} covers all additive-Gaussian-noise
models $\Cb = f(\T) + \varepsilon$ with $\varepsilon$ Gaussian
(possibly with $f$-dependent variance), where $\Ph(\T) := f(\T)$.
It also covers conditional Gaussian models whose mean and covariance
are both functions of a lower-dimensional feature map. In either
case, sufficiency follows directly from the displayed conditional
law. No Gaussian optimality claim is required.
\end{remark}

\begin{remark}[Scope beyond the Gaussian case]
Form~\eqref{eq:nonlinear-gaussian} is not exhaustive.
Theorem~\ref{thm:main} applies to any conditional $p(\Cb\mid\T)$
that factors through a sufficient statistic $\Ph$, regardless of
whether the conditional is Gaussian. The Gaussian form is
highlighted because the factorisation is easy to verify. Except in
special cases such as the jointly Gaussian linear model, the reduced
problem need not have a closed-form solution. The theorem says only
that no information outside $\Ph(\T)$ is needed.
\end{remark}

\section{Operational form under logarithmic loss}
\label{sec:operational}

The reduction has an operational form that is visible at finite
blocklength. No coding theorem is needed for the reduction itself.
Throughout this section $\Cb$ is finite, so logarithmic loss
$d(c,q):=\log(1/q(c))$ is defined for $q\in\mathcal P(\mathcal C)$.
The source $\T$ and the statistic $Z=\Ph(\T)$ remain standard Borel,
and $(\T_i,\Cb_i)_{i=1}^n$ denotes an i.i.d. block.

For positive integers $n$ and $L$, let
$D_{\T\to\Cb}^{(n)}(L)$ be the infimum of
\begin{equation*}
 \frac1n\sum_{i=1}^n \E d(\Cb_i,Q_i)
\end{equation*}
over block encoders $f:\mathcal T^n\to\{1,\ldots,L\}$ and decoders
that assign a soft reconstruction $Q_i(m)\in\mathcal P(\mathcal C)$
to each message $m$ and coordinate $i$. Define
$D_{Z\to\Cb}^{(n)}(L)$ analogously.
The operational rate--distortion function is
\begin{equation*}
R_{\T\to\Cb}(D):=
\inf\left\{R:\limsup_{n\to\infty}
D_{\T\to\Cb}^{(n)}\!\left(\lceil\exp(nR)\rceil\right)\leq D\right\},
\end{equation*}
and analogously for $Z$.

\begin{theorem}[Finite-block operational reduction]
\label{thm:operational}
Assume the hypotheses of \Cref{thm:main} and let $\Cb$ be finite. Then
\begin{equation}
D_{\T\to\Cb}^{(n)}(L)=D_{Z\to\Cb}^{(n)}(L)
\qquad\text{for every }n,L\geq1.
\label{eq:block-distortion-equal}
\end{equation}
Consequently the operational remote rate--distortion functions satisfy
\begin{equation}
R_{\T\to\Cb}(D)=R_{Z\to\Cb}(D)
\qquad\text{for every }D\geq0.
\label{eq:rd-equal}
\end{equation}
A reduced code attains the same rate and distortion on the full source
after the symbolwise preprocessing $Z_i=\Ph(\T_i)$.
\end{theorem}

\begin{proof}
Every code for $Z^n$ is a code for $\T^n$ after composition with
$\Ph^n$, so
\begin{equation}
D_{\T\to\Cb}^{(n)}(L)\leq D_{Z\to\Cb}^{(n)}(L).
\label{eq:block-easy}
\end{equation}

For the reverse inequality, fix a full-source encoder $f$ and decoder
$Q_1,\ldots,Q_n$. For $z^n\in\mathcal Z^n$ and message
$m\in\{1,\ldots,L\}$, write
\begin{equation*}
\ell(z^n,m):=
\E\!\left[\left.
   \sum_{i=1}^n d\bigl(\Cb_i,Q_i(m)\bigr)
   \right|Z^n=z^n\right].
\end{equation*}
Choose $g(z^n)$ to be the least message that minimises $\ell(z^n,m)$.
Because the message set is finite and the conditional expectations are
measurable in $z^n$, $g$ is a measurable reduced encoder.

The one-letter sufficiency chain tensorises to
$\Cb^n\leftrightarrow Z^n\leftrightarrow\T^n$. Conditional on
$Z^n=z^n$, the distortion of the original code is therefore
\begin{align*}
&\E\!\left[\left.
  \sum_{i=1}^n d\bigl(\Cb_i,Q_i(f(\T^n))\bigr)
  \right|Z^n=z^n\right] \\
&\quad=\int \ell\bigl(z^n,f(t^n)\bigr)\,
            p(dt^n\mid z^n)
\;\geq\; \ell\bigl(z^n,g(z^n)\bigr).
\end{align*}
Thus $g$ with the same decoder and the same $L$ messages has no greater
expected distortion than the original full-source code. The infima give
$D_{Z\to\Cb}^{(n)}(L)\leq D_{\T\to\Cb}^{(n)}(L)$, which together with
\eqref{eq:block-easy} proves \eqref{eq:block-distortion-equal}.

The operational rate--distortion functions allow
$L\leq \lceil\exp(nR)\rceil$ and use the usual blocklength limit.
Equality at every $(n,L)$ therefore gives \eqref{eq:rd-equal}. The
last statement is precisely the composition used in
\eqref{eq:block-easy}.
\end{proof}

When the reduced source is finite, the common operational function has
the familiar single-letter form.

\begin{corollary}[Single-letter form for a finite reduced source]
\label{cor:operational-single-letter}
If $Z$ and $\Cb$ are finite, then
\begin{equation}
R_{\T\to\Cb}(D)=R_{Z\to\Cb}(D)
=\min_{p(\Xb\mid Z)}
 \left\{\MI(\Xb;Z):H(\Cb\mid\Xb)\leq D\right\}.
\label{eq:operational-single-letter}
\end{equation}
Equivalently, the common rate--distortion function is the generalised
inverse of the common IB curve as follows.
\begin{equation}
R_{\T\to\Cb}(D)
=\inf\left\{R\geq0:
  \mathcal I_{\T,\Cb}(R)\geq H(\Cb)-D\right\}.
\label{eq:operational-inverse}
\end{equation}
\end{corollary}

\begin{proof}
For finite $Z$, remote source coding reduces to direct rate--distortion
coding with reproduction alphabet $\mathcal P(\mathcal C)$ and modified
distortion
$\widetilde d(z,q):=\E[d(\Cb,q)\mid Z=z]$
\cite{dobrushin1962information,wolf1970transmission,witsenhausen1980indirect}. The reproduction
alphabet is a continuum and logarithmic loss is unbounded, so the
finite-alphabet rate--distortion theorem is not invoked directly. The
required coding theorem is the log-loss CEO characterisation of
Courtade and
Weissman~\cite[Theorems~10 and~11]{courtade2014multiterminal}, applied
to the single encoder that observes $Z$. Their $m$-encoder statement covers
this case after the addition of an encoder that observes a constant. This encoder
satisfies the CEO conditional-independence hypothesis trivially.
It yields a minimum of
$\MI(\Xb;Z)$ over test channels $p(\Xb\mid Z)$ and reconstructions
$q=q(\Xb)$. For a fixed $\Xb$, Gibbs' inequality makes the posterior
$q(\cdot\mid\Xb)=p(\Cb\in\cdot\mid\Xb)$ optimal, with expected loss
$H(\Cb\mid\Xb)$. This proves the last expression in
\eqref{eq:operational-single-letter}. \Cref{thm:operational} gives the
two preceding equalities. Finally,
$H(\Cb\mid\Xb)=H(\Cb)-\MI(\Xb;\Cb)$, and \Cref{thm:main}(i) turns the
constraint into \eqref{eq:operational-inverse}.
\end{proof}

The point of \Cref{thm:operational} is stronger than the
single-letter identity. A sufficient-statistic front end loses nothing
at any finite blocklength or message budget. The asymptotic
rate--distortion equality is not an artefact of convexification or a
particular coding theorem.

\section{Certifying a task-sufficient front end}
\label{sec:examples}

The hypothesis of Theorem~\ref{thm:main} is a property of a
summary--task pair, not of a summary in isolation. Its conditional-law
form is
\begin{equation}
p(\Cb\in A\mid\T=t)
=p(\Cb\in A\mid Z=\Ph(t))
\quad\text{for every measurable }A,\text{ almost surely}.
\label{eq:task-interface}
\end{equation}
Under~\eqref{eq:task-interface}, every point of the full-source IB curve
is attainable from the reduced source, every reduced optimum pulls back
to a full-source optimum, and any full-source encoder that retains
within-fibre information pays the exact excess rate
$\MI(\Xb;\T\mid Z)$ with no gain in relevance. Four canonical settings
certify the condition exactly.

\paragraph{Posterior statistic}
The conditional law itself is always sufficient. For any chosen task
variable $\Cb$, the statistic $Z=p(\Cb\in\cdot\mid\T)$ defines a
random probability measure on $(\mathcal{C},\mathcal{B}_\Cb)$ and satisfies
\eqref{eq:task-interface}. For a finite label set this is the posterior
probability vector. For general $\Cb$ it is measure-valued and
typically infinite-dimensional, so a simpler sufficient statistic is
preferable when one exists. A hard predicted label is sufficient only in the
special case that the full posterior is constant on its label fibres.

\paragraph{Structured conditionals}
If $p(\Cb\mid\T=t)$ depends on $t$ only through a finite-dimensional
parameter $\eta(t)$, as for an exponential-family conditional with
natural parameter $\eta(t)$, then $Z=\eta(\T)$ satisfies
\eqref{eq:task-interface}. Regression with independent additive noise
is the special case in which $\eta$ is the regression function. For
jointly Gaussian pairs the conditional mean is such a statistic, and
\Cref{sec:gaussian} develops the consequences.

\paragraph{Deterministic tasks}
If $\Cb=g(\T)$ almost surely, condition \eqref{eq:task-interface}
reduces to the existence of a measurable $h$ such that $g=h\circ\Ph$
almost surely. A summary is exactly sufficient precisely when the task
function is constant on each fibre of $\Ph$. Variation of the task
value within a fibre is precisely the case in which an exact reduction
cannot be asserted. The degenerate IB behaviour of this regime is the
subject of \Cref{cor:kolchinsky}.

\paragraph{Sensor sources with irrelevant nuisance}
Let $\T=(S,W)$ with $\Cb$ conditionally independent of $W$ given $S$.
The coordinate projection $\Ph(s,w)=s$ then satisfies
\eqref{eq:task-interface}. The nuisance component can add rate to a
representation but cannot add relevance. This is the projection form
in which Harremo\"{e}s and
Tishby recorded their original observation, and the exact binary
example of \Cref{sec:numerics} instantiates it. Each copied nuisance bit
costs one bit of rate and adds nothing to the relevance term.

Condition~\eqref{eq:task-interface} is deliberately task-relative. A
summary that satisfies it for one task variable may fail it for a
different target defined on the same source, so sufficiency must be
checked separately for each task. If only approximate sufficiency is
known, the exact equalities of this paper are not claimed.

\section{An exact finite example}
\label{sec:numerics}

The reduction can be seen without estimation or numerical optimisation.
Let $Z\sim\operatorname{Bernoulli}(1/2)$, let
$W\sim\operatorname{Uniform}\{0,1\}^d$ be independent of $Z$, and set
$\T:=(Z,W)$. Let
\begin{equation*}
\Cb=Z\oplus N,
\qquad N\sim\operatorname{Bernoulli}(\varepsilon),
\end{equation*}
where $N$ is independent of $(Z,W)$. Then
$\Cb\leftrightarrow Z\leftrightarrow\T$, so the projection
$\Ph(z,w)=z$ is sufficient. The full source has $2^{d+1}$ states,
whereas the reduced source has two.

For a reduced binary encoder
\begin{equation*}
\Xb=Z\oplus E,
\qquad E\sim\operatorname{Bernoulli}(\delta),
\end{equation*}
with $E$ independent of $(Z,W,N)$, direct calculation gives
\begin{align}
\MI(\Xb;\T)=\MI(\Xb;Z)
  &=\log 2-h_2(\delta), \\
\MI(\Xb;\Cb)
  &=\log 2-h_2(\varepsilon\star\delta),
\end{align}
where $h_2(u):=-u\log u-(1-u)\log(1-u)$ and
$\varepsilon\star\delta:=\varepsilon+\delta-2\varepsilon\delta$.
Neither quantity depends on $d$.

Now consider the deliberately wasteful full-source representation
$Y:=(\Xb,W)$. Since $W$ is independent of $(Z,\Cb,\Xb)$,
\begin{align*}
\MI(Y;\Cb)&=\MI(\Xb;\Cb),\\
\MI(Y;\T)&=\MI(\Xb;Z)+d\log 2.
\end{align*}
The fibre average of $Y$ retains the same binary encoder $\Xb$ and
replaces the copied nuisance coordinate by an independent uniform
one. Its relevance is unchanged and its rate is smaller by exactly
\begin{equation*}
\MI(Y;\T\mid Z)=H(W)=d\log 2.
\end{equation*}
This is \eqref{eq:fibre-lagrangian} in arithmetic form. Information
spent on variation within a fibre of the sufficient statistic can
increase compression cost but cannot increase relevance.

\section{Discussion}
\label{sec:discussion}

\paragraph{Scope of the reduction}
The reduction applies whenever $p(\Cb\mid\T)$ factors through a
statistic $\Ph$.
\Cref{sec:examples} collects the canonical certifications. The
theorem is useful when such a statistic is known and materially
simpler than the source. The examples also show why the target must
be named. ``Sufficient summary'' is incomplete unless the task
variable is also specified.

\paragraph{Relation to prior work}
The fact that sufficient statistics preserve mutual information is
classical~\cite{cover2006elements}. The closest prior statement
specifically about IB is that of Harremo\"{e}s and
Tishby~\cite[Section~III, p.~568]{harremoes2007ib}. They consider the
projection-form case $X = (X_1, X_2)$ with
$X_1 \perp Y \mid X_2$ (their $X$, $Y$ are our $\T$, $\Cb$). They
write that ``we should compare the bottleneck problem $X \to Y$ with the
bottleneck problem $X_2 \to Y$ and show that they have the same
rate distortion function'' and support the claim with a sketched
joint-extension argument. Their setting is explicitly
finite-alphabet (``for simplicity we shall assume that $\mathbb{A}$
and $\mathbb{B}$ are finite'', \cite[p.~566]{harremoes2007ib}) with
the parenthetical note that ``this sufficiency result on the input
side holds for any distortion measure''
\cite[p.~568]{harremoes2007ib}. The statement is brief and not
developed further. They give no minimiser correspondence or
Lagrangian-level equivalence, and no corollaries are
drawn. The observation does not appear to have been picked up in
the subsequent IB literature, and to our knowledge has not been
elevated to a theorem in the intervening two decades.
\Cref{thm:main} completes the program implicit in that observation.
We prove three extensions on standard Borel spaces under essentially no
regularity beyond $\MI(\T;\Cb) < \infty$. First, curve equality holds
for an arbitrary measurable sufficient statistic
$\Ph:\mathcal{T}\to\mathcal{Z}$, not only for coordinate projections.
Second, the Lagrangian is preserved at every $\beta$. Third, there is an
explicit pullback
correspondence~\eqref{eq:pullback} between minimisers of the two
problems.

There is also a useful decision-theoretic reading. As channels
from $\Cb$, $Z=\Ph(\T)$ is a deterministic garbling of $\T$, while
the sufficiency chain $\Cb\leftrightarrow Z\leftrightarrow\T$
supplies the reverse garbling $p(\T\mid Z)$. Thus $\T$ and $Z$ are
Blackwell-equivalent experiments about $\Cb$
\cite{blackwell1953equivalent}. Blackwell equivalence guarantees equal
value in every decision problem, but it does not by itself equate the
ordinary IB rate coordinates $\MI(\Xb;\T)$ and $\MI(\Xb;Z)$.
Equivalence constrains the posteriors an observer can reach, not the
rate an encoder must spend to reach them.
The fibre-average identity supplies that rate-sensitive conclusion
and identifies the exact difference $\MI(\Xb;\T\mid Z)$.

Kolchinsky's redundancy bottleneck~\cite{kolchinsky2024redundancy}
uses Blackwell order in a distinct multi-source problem from partial
information decomposition. It relaxes the zero-leakage formulation of
Blackwell redundancy. The relaxed form permits conditional information
about which source supplied the representation. Its compression coordinate is source-identity
leakage rather than the ordinary IB source rate. It does not consider
replacement of one source by a sufficient statistic or establish
curve, Lagrangian, minimiser, or operational equivalence for the
standard IB. The two results are complementary. The redundancy
bottleneck turns Blackwell comparability across sources into a new
tradeoff, whereas \Cref{thm:main} turns exact Blackwell equivalence of
a source and its task-sufficient statistic into an invariance of the
ordinary IB.

The Gaussian IB of~\cite{chechik2005gaussian} is then an
immediate corollary (\Cref{cor:chechik}). The deterministic-$\Cb$
pathologies documented by Kolchinsky et
al.~\cite{kolchinsky2019caveats} are recovered as the degenerate
case treated in \Cref{cor:kolchinsky}. The algorithmic
consequence is that a finite reduced problem depends on the alphabet
of $\Ph(\T)$ rather than the alphabet of $\T$
(\Cref{cor:computation}). Related ideas appear implicitly in the
exponential-family IB of~\cite{painsky2018gaussianlower}, in the
agglomerative Blahut--Arimoto reduction
of~\cite{slonim2000agglomerative}, and in the minimal-sufficient representation
literature~\cite{achille2018emergence}. The last of these imposes
sufficiency on the representation side and addresses a different
question. The agglomerative method clusters the alphabet of a discrete
source before the Blahut--Arimoto iteration. It is an approximate,
data-driven counterpart of the reduction and is exact precisely when
merged symbols share the same conditional $p(\Cb\mid t)$, which means
the merge map is sufficient. Its bottom-up merge criterion is a
marginal-weighted Jensen--Shannon divergence between the merged
conditionals, and it vanishes exactly under that condition.
Several adjacent lines of IB work address
distinct questions. Shamir, Sabato, and
Tishby~\cite{shamir2010learning} derive finite-sample generalisation
bounds on the representation-side estimator.
\Cref{thm:main} is a population-level source-side identity,
orthogonal to their finite-sample analysis. Hsu, Asoodeh,
Salamatian, and Calmon~\cite{hsu2018generalizing} generalise the
bottleneck to $f$-information functionals under a matched-channel
hypothesis. The reduction here is a coordinate identity within
standard mutual information, without $f$-generalisation. Asoodeh
and Calmon~\cite{asoodeh2020bottleneck} give an
estimation-theoretic view that connects IB and privacy funnel to
hypothesis testing and noisy source coding. This is a functional
recasting rather than a source-side reduction. A canonical survey of IB variants
that covers distributed IB and its CEO interpretation is Zaidi,
Estella-Aguerri, and Shamai~\cite{zaidi2020ibsurvey}. The proofs above
are deliberately elementary. They use one pullback, one conditional
average, and the mutual-information chain rule. We
regard this as a feature of the result rather than a limitation.
The reduction was available at this cost for two decades. The
contribution is the identification of the objects that make it an
exact theorem, together with its consequences.

\paragraph{Limitations}
Theorem~\ref{thm:main} is a population result. In applications,
$\Ph$ must either be known a priori or estimated, and the reduced
problem then inherits the usual finite-sample error in
$\widehat{p}(\Ph,\Cb)$. When $\Ph$ is unknown, the theorem is best
read as a conditional statement, not as an algorithm that discovers
a statistic. One must first establish sufficiency and then solve IB on the reduced
pair. Approximate or estimated sufficiency requires a separate error
analysis. In particular, the theorem does not turn an approximately
task-sufficient summary into an exact reduction.

\paragraph{Future directions}
Natural extensions are a finite-sample analysis of the reduced
problem with estimated $\widehat{p}(\Ph,\Cb)$, a stability analysis
under approximate sufficiency in which $\MI(\Cb;\T\mid Z)$ is small
but nonzero, application to multi-view and conditional IB variants,
and the use of the reduction as a preconditioning step inside neural
IB methods.

\appendix

\section{Notation summary}
\label{app:notation}

\begin{table}[!ht]
\centering
\begin{tabular}{ll}
\toprule
Symbol & Meaning \\
\midrule
$\T$, $\Cb$, $\Xb$ & source, relevance variable, representation \\
$\Ph$ & sufficient statistic map $\Ph:\mathcal{T}\to\mathcal{Z}$ \\
$\mathcal{I}_{\T,\Cb}(R)$ & IB information curve \\
$\mathcal{L}_\beta^{\T}$ & IB Lagrangian infimum at trade-off $\beta$ \\
$\MI(A;B)$ & mutual information between $A$ and $B$ \\
$\MI(A;B\mid C)$ & conditional mutual information \\
$p(\Xb\mid\T)$ & encoder (stochastic kernel) \\
$\R^d$ & $d$-dimensional Euclidean space \\
$\Sigma_{AB}$ & cross-covariance of $A$ and $B$ \\
$r$ & rank of $\Sigma_{\Cb\T}\Sigma_{\T\T}^{-1}$ \\
\bottomrule
\end{tabular}
\caption{Notation used throughout the paper.}
\label{tab:notation}
\end{table}

\section{Measurability details}
\label{app:measurability}

Let $(\mathcal{T},\mathscr{F}_\T)$, $(\mathcal{Z},\mathscr{F}_\Z)$,
and $(\mathcal{X},\mathscr{F}_\Xb)$ be standard Borel spaces, and let
$\Ph : \mathcal{T} \to \mathcal{Z}$ be measurable. The conditional
distributions $p(\Cb\mid\T=t)$ and $p(\Xb\mid\T=t)$ are taken as
regular conditional probabilities, which exist on standard Borel
spaces~\cite{kallenberg2002foundations}. See
also~\cite{dellacherie1978probabilities,bogachev2007measure} for
disintegration and regular conditional distributions in fuller
generality. The integrals
\[
q(\Xb\in A\mid\Ph=z)
\;=\;
\int p(\Xb\in A\mid\T=t)\,p(\T\in dt\mid\Ph=z)
\]
are well-defined as disintegrations of $p(\T,\Xb)$ along $\Ph$,
and the resulting kernel $q(\Xb\mid\Ph)$ is measurable in $z$ by
standard disintegration theory.

\bibliographystyle{plain}
\bibliography{references}

\end{document}